\documentclass[12pt]{article}

\usepackage{arxiv}

\usepackage[utf8]{inputenc} % allow utf-8 input
\usepackage[T1]{fontenc}    % use 8-bit T1 fonts
\usepackage{hyperref}       % hyperlinks
\usepackage{url}            % simple URL typesetting
\usepackage{booktabs}       % professional-quality tables
\usepackage{amsfonts}       % blackboard math symbols
\usepackage{nicefrac}       % compact symbols for 1/2, etc.
\usepackage{microtype}      % microtypography
\usepackage{lipsum}
\usepackage{graphicx}
\usepackage{amssymb,amsfonts,amsmath,amsthm}
\usepackage{algorithm}
\usepackage{algpseudocode}
\usepackage{subcaption}
\usepackage{textcomp}
\usepackage{xcolor}
\usepackage[english]{babel}
\DeclareMathOperator*{\argmax}{arg\,max}
\newtheorem{lemma}{Lemma}
\usepackage{authblk}

\title{\textsc{Lum\'awig}: \\An Efficient Algorithm for Dimension Zero Bottleneck Distance Computation in Topological Data Analysis}

\author[$\star\dagger$]{Paul Samuel Ignacio}
\author[$\dagger$]{Jay-Anne B. Bulauan}
\author[$\ddag$]{David Uminsky}

\affil[$\dagger$]{University of the Philippines Baguio, Baguio City, Philippines 2600}
\affil[$\ddag$]{University of Chicago, Chicago, Il, United States 60637}
\affil[$\star$]{Correspondence: ppignacio@up.edu.ph}

\begin{document}
\date{}
\maketitle

\begin{abstract}
Stability of persistence diagrams under slight perturbations is a key characteristic behind the validity and growing popularity of topological data analysis in exploring real-world data. Central to this stability is the use of Bottleneck distance which entails matching points between diagrams. Use of this metric in practical studies has, however, been few and sparingly because of the computational obstruction, especially in dimension zero where the computational cost explodes with the growth of data size. We present \textsc{Lum\'awig}, a novel efficient algorithm to compute dimension zero bottleneck distance between two persistent diagrams which runs significantly faster and provides significantly sharper approximates with respect to the output of the original algorithm than any other available algorithm. We bypass the overwhelming matching problem in previous implementations of the bottleneck distance, and prove that the zero dimensional bottleneck distance can be recovered from a very small number of matching cases. We show that \textsc{Lum\'awig} generally enjoys linear complexity as shown by empirical tests. We also present an application that leverages dimension zero persistence diagrams and the bottleneck distance to produce features for classification tasks.
\end{abstract}

% keywords can be removed

\section{Introduction}
Topological data analysis (TDA) has gathered significant interest from a wide range of researchers because of its novel approach and use of classical tools from algebraic topology for extracting descriptive features from data. Succinctly, topological data analysis captures and records the persistence \cite{zomorodian1, zomorodian2} of algebraically computable topological signatures, and regards it as a measure of significance for different features embedded in the structure of data. Meriting the growing popularity for this approach, and central to its relevance and viability in interrogating real-world data, is its stability under slight perturbations -- small discrepancies between measurements within data lead to small differences in the recorded persistence of features. This cornerstone stability result \cite{cohen} relies on classic bottleneck matchings to evaluate, measure, and bound changes between two records of feature persistence. These records, called \emph{persistence diagrams}, are a collection of points in the extended plane where the coordinates represent the birth and death times of the recorded features. 

Given two persistence diagrams $X$ and $Y$, the \emph{bottleneck distance} between them is defined as
$$
d_B(X,Y) = \inf_{\phi} \sup_{x\in X} ||x-\phi(x)||_{\infty}
$$
where the infimum is taken over all bijections $\phi:X\sqcup \Delta \to Y\sqcup \Delta$ and $\Delta$ is the diagonal. In general terms, the bottleneck distance measures the cost to transform one diagram to another. The first, and for a long time the only, publicly available implementation of the bottleneck distance for persistence diagrams is in the library \textsc{Dionysus}, released in 2010, by Morozov \cite{dio}. This implementation uses a variant of the Hungarian algorithm \cite{munkres} for the assignment problem. 

Understandably, because of the overwhelming matching step in the computation, this first implementation of the bottleneck distance between two persistence diagrams was considerably slow by practical standards. Consequently, while the theoretical side of topological data analysis has made extensive use of the bottleneck distance for advancing the theory \cite{zigzag,samir2,ignacio}, first computational uses have been few and sparingly. Some notable examples include applications to classification of hepatic lesions \cite{lesions}, and analysis of time-series data \cite{cvpr} and simulated hippocampal networks \cite{samir}. Most applications of TDA, instead, tap into persistence-based topological features via another class of objects, called \emph{persistence landscapes} \cite{bubenik}, that record the persistence of features as a function, thus affording access to desirable properties of the underlying function space. A major motivation for this detour to landscapes is the ability to generate topological summaries that are compatible to classical tools in statistics, and even machine learning.

 In 2017, Morozov et al. \cite{Kerber} provided an improved implementation of the bottleneck distance in the library \textsc{Hera} by exploiting geometry. Their approach follows closely the work of Efrat et al. \cite{Efrat}. For the sets $X_0$ and $Y_0$ of orthogonal projections on the diagonal $\Delta$ of points respectively from $X$ and $Y$, and the sets $U = X \cup Y_0$ and $V = X_0\cup Y$, they consider the weighted complete bipartite graph $G = (U\sqcup V, U\times V, w)$ where $w:U\times V\to \mathbb{R}_{\geq 0}$ is given by $$w(u,v) = \begin{cases}
 ||u-v||_{\infty} & \mbox{ if } u \in X \mbox{ or } v\in Y\\
 0 & \mbox{otherwise}.
  \end{cases}
 $$
With this, the bottleneck computation problem can be recast in the following manner: if $G[r]$ is the subgraph of $G$ having all edges $e$ of weight $w(e)\leq r$, then the bottleneck distance of $G$ is the minimal value $r$ such that $G[r]$ contains a perfect matching. Hence the bottleneck distance can be recovered by combining a binary search on the edge weights of $G$ with a test for a perfect matching. For the matching step, they augment the Hopcroft-Karp algorithm \cite{hopcroft} by appealing to a near-neighbor data structure (a k-d tree) to search for the best candidate pair for a query point, pruning from the search the subtrees (and hence all other candidates within them) whose enclosing box is further away from the query than the current best candidate. This circumvents the overwhelming matching problem by significantly shrinking down the combination pool to retrieve the best matching. To approximate complexity, they fit curves of the form $cn^{\alpha}$ and found a best fit with $\alpha = 1.4$. This translates to speed-up from \textsc{Dionysus} already by a factor of 400 on diagrams with 2,800 points, and opened opportunities for several works that examine larger \cite{mnist} or more complex \cite{traffic, lung} data sets.

 We take inspiration from this idea of exploiting the geometry of persistence diagrams to extract computational speed-up. By considering persistence diagrams whose components are assumed to be born at the beginning of the filtration, we can approach the problem via a different framework, birthing a new efficient algorithm for computing the bottleneck distance. The key idea is to begin with a specific initial bijection that one can methodically modify to optimize the norm between matched points. This process allows us to identify all possible instances where the bottleneck matching is achieved, and the exact value for the bottleneck distance, significantly bypassing the overwhelming matching step in previous implementations.
 
 Partly in keeping with nomenclature traditions in this area of TDA, we name this algorithm \textsc{Lum\'awig} as a nod to a deity in the northern Philippines, where the algorithm was developed. \textsc{Lum\'awig} is significantly faster than the state-of-the-art and provides significantly sharper approximates with respect to the output of the original algorithm than any other available algorithm. We benchmark \textsc{Lum\'awig} against all available algorithms in terms of running time and accuracy. 

Our motivation for this work is to clear the computational obstruction in the use of bottleneck distance in applications. In the Filipino language, \textsc{Lum\'awig} also means to extend, broaden, or expand. Our hope is that this contribution will serve as a catalyst in the further development of the theory that leverages persistence diagrams and the bottleneck distance similar to what has been achieved for persistence landscapes, and will usher in a new era of integrating TDA into the science of big data. As a proof of concept, we use \textsc{Lum\'awig} to generate features for the classification of digit images from the MNIST data set.

\iffalse
The rest of the paper is organized as follows. We first detail the bypassing step in our approach, show how the bottleneck distance can be recovered from a very small number of cases, and provide the pseudo code for the \textsc{Lum\'awig} algorithm. Next, we perform some bechmarking of the performance of \textsc{Lum\'awig} against all publicly available implementations of the bottleneck distance, and provide an experimental analysis of its complexity. We then present an application of \textsc{Lum\'awig} in feature extraction for digit classification via machine learning. Finally, we end with some discussions and conclusions, including future plans, about our work. 
\fi

\section{Bypassing matchings}
We propose to bypass the overwhelming matching problem in the computation of 0-dimensional bottleneck distance by showing that the value produced by the bottleneck distance formula can be recovered by considering only a few cases. We will show that these cases naturally come up in the process of minimizing the output of the norm. 

We first note that for most practical applications to data analysis of 0-dimensional persistence diagrams, where all components are assumed to be born at the beginning of the filtration for persistent homology, all non-trivial points lie in the vertical axis (or equivalently for persistence barcodes, all bars begin at time $t=0$). Hence, in this case, if $\delta_x$, and $\delta_{\phi(x)}$ are the death times respectively for $x$ and its matched point $\phi(x)$, we have that
\begin{equation}\label{normdef}
||x - \phi(x) ||_{\infty} =  \begin{cases}
\max(\delta_x,\delta_{\phi(x)})/2 &\mbox{ if } \phi(x) \in \Delta\\
|\delta_x-\delta_{\phi(x)}| &\mbox{ otherwise.}
\end{cases}
\end{equation}
This suggests that while it is natural to do a point-to-point matching between diagrams, there are cases when we are better off matching a point to the diagonal. For a point $x\in X$ and $\phi(x)\in Y$, this happens precisely when 
\begin{equation}\label{ineq1}
\max(\delta_{x},\delta_{\phi(x)})>2\min(\delta_{\phi(x)},\delta_{x}).
\end{equation}
See Figure \ref{bott}a. Therefore, unless (\ref{ineq1}) is satisfied, it is our priority to match a non-trivial point in a diagram $X$ with a non-trivial point in another diagram. This supports the interpretation that the bottleneck distance is the cost of transforming one diagram to another. 
\begin{figure}
\includegraphics[width =\textwidth, height = 0.24\textwidth]{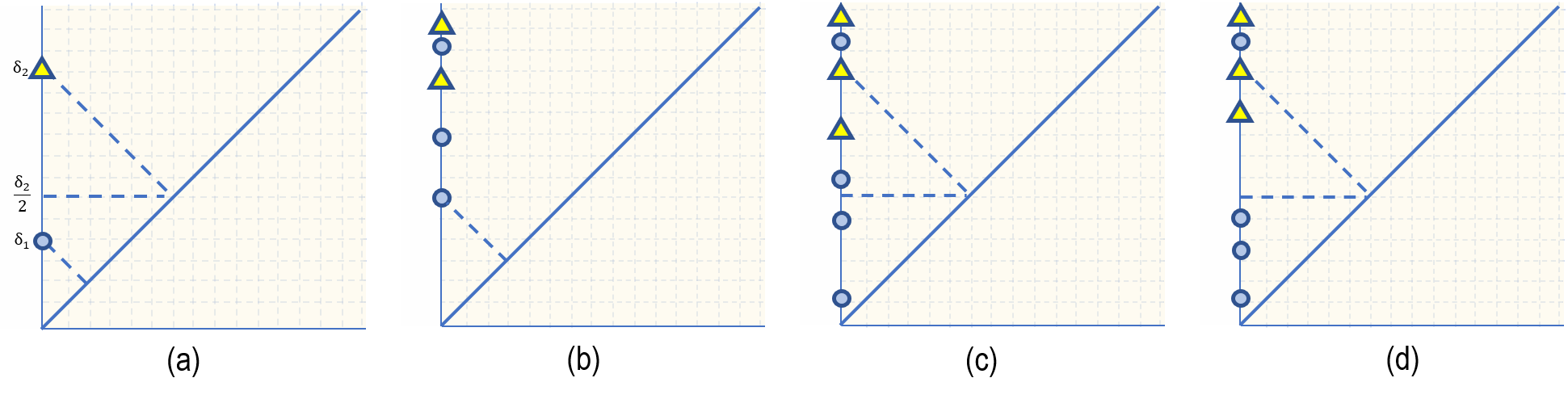}
\caption{Matching points between persistence diagrams.}
\label{bott}
\end{figure}

We are now ready to present our algorithm for computing 0-dimensional bottleneck distance between two persistence diagrams. We first induce and ordering of the death times in both diagrams and define a bijection that we can methodically modify to optimize the norm between matched points and recover the desired matching that achieves the bottleneck distance. The proof of Lemma \ref{lem1} provides the basic argument that allows us to bypass the overwhelming matching problem. Lemma \ref{lem2} proceeds in the same manner and 
identifies all other possible instances where the bottleneck matching is achieved, and the exact bottleneck distance in each case.  

Let $X$ and $Y$ be two 0-dimensional persistence diagrams whose death time entries are arranged from largest to smallest. Equivalently, $X$ and $Y$ can be thought of as persistence barcodes whose bars are arranged from longest to shortest. Without loss of generality, assume that $X$ has at most as many points as $Y$ has. We remark that this pre-processing is equivalent to considering the bijection $\phi$ that matches points between $X$ and $Y$ according to the relative ranking of death times from largest to smallest, and where unmatched points in $Y$ are matched to the diagonal. Let $N = length(X)$ and define 
$$Z = [z_i]_1^{length(Y)} \mbox{ where } z_i = \begin{cases}
|x_i-y_i| &\mbox{ if } i \leq N\\
y_i/2 &\mbox{ otherwise}
\end{cases}$$
and $l = \argmax(Z)$. 

\begin{lemma}\label{lem1}
Let $X$, $Y$, $Z$, $N$ and $\phi$ be defined as above. If $N<length(Y)$ and $\max(Z)\leq y_{N+1}/2$, then 
$$
d_B(X,Y) = y_{N+1}/2
$$
where $y_{N+1}$ is the largest death time of a point in $Y$ matched to the diagonal.
\end{lemma}
\begin{proof}
For the bijection $\phi$ corresponding to the pre-processing described above, it follows that
$$
\max_{x\in X} ||x-\phi(x)||_{\infty} = y_{N+1}/2.
$$
To see why $\phi$ achieves the infimum over all bijections between $X$ and $Y$, note that any other bijection $\psi$ produces a death time for a point in $Y$ matched to the diagonal that is at least as large as $y_{N+1}/2$. Therefore $\max_{x\in X} ||x-\phi(x)||_{\infty}\leq \max_{x\in X} ||x-\psi(x)||_{\infty}$. See Figure\ref{bott}b.
\end{proof}

\begin{lemma}
Let $X$, $Y$, $Z$, $N$, $l$ and $\phi$ be defined as above, and let $\zeta$ be the second largest entry of $Z$.
\begin{enumerate}
\item If $\max(Z)\leq \max(x_l,y_l)/2$, then $d_B(X,Y) = \max(Z).$
\item If $\zeta < \max(x_l,y_l)/2<\max(Z)$, then $d_B(X,Y) = \max(x_l,y_l)/2.$
\item If $\zeta \geq \max(x_l,y_l)/2$ and $m\geq l$ for every $m$ such that $z_m\geq \max(x_l,y_l)/2$, then $d_B(X,Y) = \max(x_l,y_l)/2.$
\item If $\zeta \geq \max(x_l,y_l)/2$ and there exists $m < l$ such that $z_m\geq \max(x_l,y_l)/2$, then there exists a bijection $\tau$ between $X$ and $Y$ such that one of the three preceding cases holds and where 
$$\max ||x - \tau(x) ||_{\infty} < \max ||x - \phi(x) ||_{\infty}.$$
\end{enumerate}
\label{lem2}
\end{lemma}
\begin{proof}
\begin{enumerate}
\item It follows from our remark immediately after (\ref{normdef}) that 
$$\max ||x - \phi(x) ||_{\infty} = \max(Z) \leq \max(x_l,y_l)/2 = \max ||x - \phi^{\prime}(x) ||_{\infty}$$
where $\phi^{\prime}$ is the bijection that matches both $x_l$ and $y_l$ to the diagonal, and coincides with $\phi$ otherwise. For any other bijection $\psi$, if $x^{\prime}\in X$ such that $|x^{\prime} - \psi(x^{\prime})|$ is maximum among all non-trivial matchings, either $\max(Z)\leq |x^{\prime} - \psi(x^{\prime})|$, or $\max(x_l,y_l)\leq \max(x^{\prime},\psi(x^{\prime}))$. See Figure \ref{bott}c. If $N< length(Y)$, then a similar argument as that in Lemma \ref{lem1} holds. The conclusion now follows. 
\item In this case, the same bijection $\phi^{\prime}$ in the previous case yields
$$\max ||x - \phi^{\prime}(x) ||_{\infty} =  \max(x_l,y_l)/2 < \max(Z) = \max ||x - \phi(x) ||_{\infty}.$$
The same argument in the previous case holds for any other bijection $\psi$. Hence, the inequality above implies the conclusion.
\item For the bijection $\phi^{\prime\prime}$ that sends $x_m$ and $y_m$ to the diagonal for all such $m$, and coincides with $\phi$ otherwise (see Figure\ref{bott}d), we have that 
$$\max ||x - \phi^{\prime\prime}(x) ||_{\infty} =  \max(x_l,y_l)/2 < \max(Z) = \max ||x - \phi(x) ||_{\infty}.$$
Again, since the same argument in the first case holds for any other bijection $\psi$, the previous inequality implies the conclusion.
\item Define the bijection $\tau$ that sends $x_j$ and $y_j$ to the diagonal for all $j \geq l$, and coincides with $\phi$ otherwise. Then we have that
$$\max ||x - \tau(x) ||_{\infty} < \max(Z) = \max ||x - \phi(x) ||_{\infty},$$
Moreover, note that $\max ||x - \tau(x) ||_{\infty}$ depends only on $||x - \tau(x)||_{\infty}$ for non-trivially matched $x$ and $\tau(x)$. Therefore, we can consider only the subsets $X^{\prime}$ and $Y^{\prime}$ respectively of $X$ and $Y$ whose points are non-trivially matched by $\tau$. In this case $length(X^{\prime}) = length(Y^{\prime})$ and one of the three previous cases above holds.
\end{enumerate}
The proof is now complete.
\end{proof}

The two Lemmas above provide the theoretical basis for the bypass approach of the \textsc{Lum\'awig} algorithm. Together, they take advantage of the specific form of dimension zero persistence diagrams being considered, and the methodical approach to optimize norms induced by a specific matching. The complete pseudo code for the algorithm is given below.

\begin{algorithm}[h]
\caption{\textsc{Lum\'awig} algorithm for computing 0-dimensional bottleneck distance between two persistence diagrams}
\label{algo}
\begin{algorithmic}[1]
\State \textbf{Input: } Two dimension zero persistence diagrams $X$ and $Y$ such that $X\neq Y$ and where $X$ has fewer than or as many points as $Y$.
\State \textbf{Output: } The bottleneck distance between $X$ and $Y$.

\State Initialization\; $d \leftarrow 0$, $X\leftarrow$ death times of points from $X$ sorted from largest to smallest, $Y\leftarrow$ death times of points from $Y$ sorted from largest to smallest, $N = length(X)$, $Z\leftarrow$ vector $[z_i:=|x_i-y_i|]_{1}^N$, $l = \argmax(Z)$, $d_{temp} = \max(Z)$

		\If {$length(X) \neq length(Y) $ and $ d_{temp} < y_{N+1}/2$}
			\State $d = (y_{N+1})/2$;
		\Else
			\While {$length(Z)>1$} 
				\If {$\text{Second largest entry of } Z < \max(x_l,y_l)/2<d_{temp}$}
					\State $d = \max(x_l,y_l)/2$
					\State $break$
				\ElsIf {$\text{Second largest entry of } Z \geq \max(x_l,y_l)/2$}
					\If{ For every $m$ for which $z_m\geq \max(x_l,y_l)/2$, $m\geq l$ }
						\State $d = \max(x_l,y_l)/2$
						\State $break$
					\Else
						\State Trim off all $z_m, x_m, y_m$ for $m\geq l$; update $l$ and $d_{temp}$
						\If{$length(Z)=1$}
							\State $d = \min(d_{temp}, \max(x_l,y_l)/2)$
							\State $break$
						\EndIf
					\EndIf
				\Else
					\State $d = d_{temp}$
					\State $break$
				\EndIf
			\EndWhile	
	   \EndIf	
\end{algorithmic}
\end{algorithm}

\section{Benchmarking}
We perform two stages of benchmarking against other publicly available implementations of the bottleneck distance. The first stage is in terms of computational running time and relative difference with respect to the original algorithm implemented in the \textsc{Dionysus} library included in the \textsc{R} package \textsc{TDA} \cite{TDAR}. This stage involves persistence diagrams with as many as 900 points and highlights the computational obstructions with current implementations of the bottleneck distance.

The second stage is also done in terms of running time, but the relative difference is with respect to the \textsc{R} implementation of \textsc{Lum\'awig}. Only the faster implementations are considered in this stage as it involves persistence diagrams with as many as 30,000 points.

\subsection{Benchmarking against all available algorithms}
Figure \ref{oldrun} shows the running time (in seconds) of four algorithms for computing the bottleneck distance between two persistence diagrams: the \textsc{Dionysus} implementation  (in \textsc{R}), the current state-of-the-art \textsc{Hera} (implemented in \textsc{C++} and wrapped in \textsc{Python}), a new implementation \textsc{Persim} (\textsc{Python}) \cite{persim}, and \textsc{Lum\'awig}, implemented both in \textsc{R} and \textsc{Python}. The benchmarking is done by first simulating 100 0-dimensional persistence diagrams with 50 points. Each dimension zero persistence diagram is simulated using a set of positive numbers as death times uniformly chosen from a range twice as wide as the number of points. We pair each diagram with another simulated diagram not necessarily having the same number of points\footnote{The second set of diagrams have as much as $80\%$ more or fewer points.}, then compute the bottleneck distance (up to 10 decimal places) between the pair using the bottleneck implementations above. The running time of each algorithm is recorded, and the distribution summary of 100 run times for each algorithm is plotted out as a boxplot. For \textsc{Hera}, we follow the experimental setup from \cite{Kerber} and set $\delta = 0.01$. We repeat this process while increasing the number of points in the base persistence diagram by 50 until we reach 500 points for each base persistence diagram.
\begin{figure}
  \centering
\includegraphics[width = \textwidth, height = 0.3\textwidth]{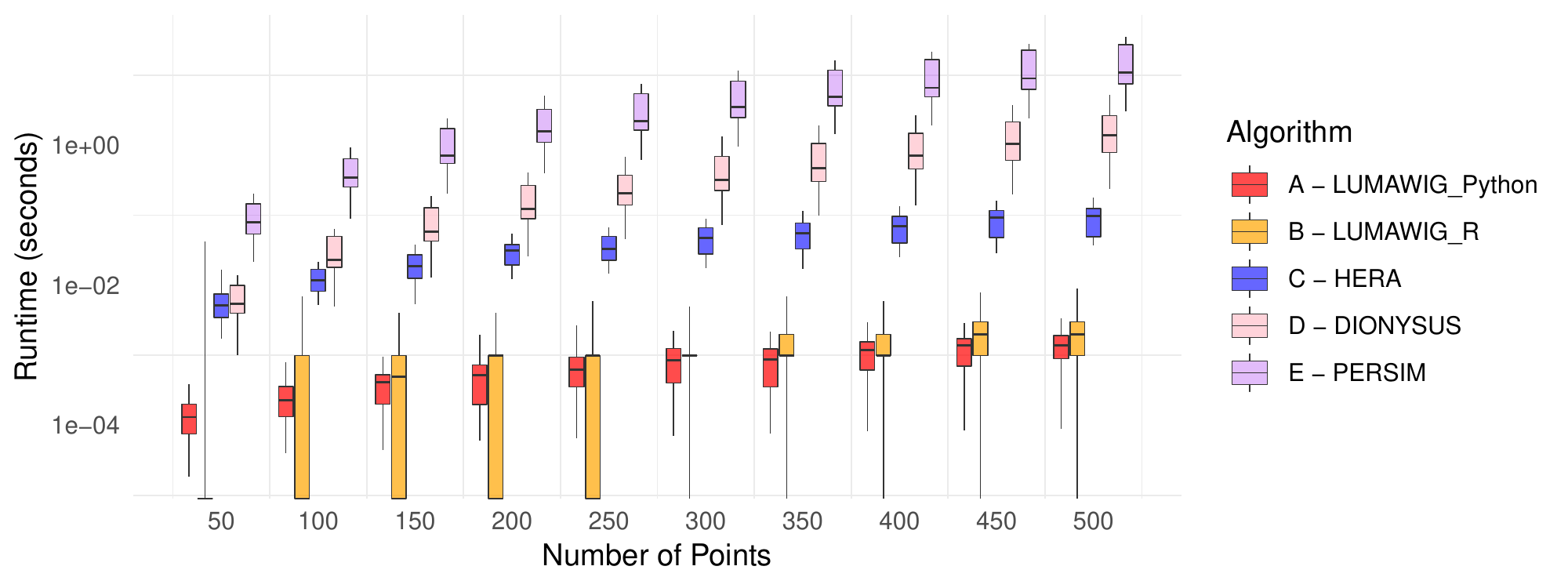}
\caption{Boxplots of running times (seconds in $\log$ scale) from different algorithms.}
\label{oldrun}
\end{figure}

\begin{figure}[h]
\begin{subfigure}{0.5\textwidth}
  \centering
\includegraphics[width =  \textwidth, height = 0.3\textwidth]{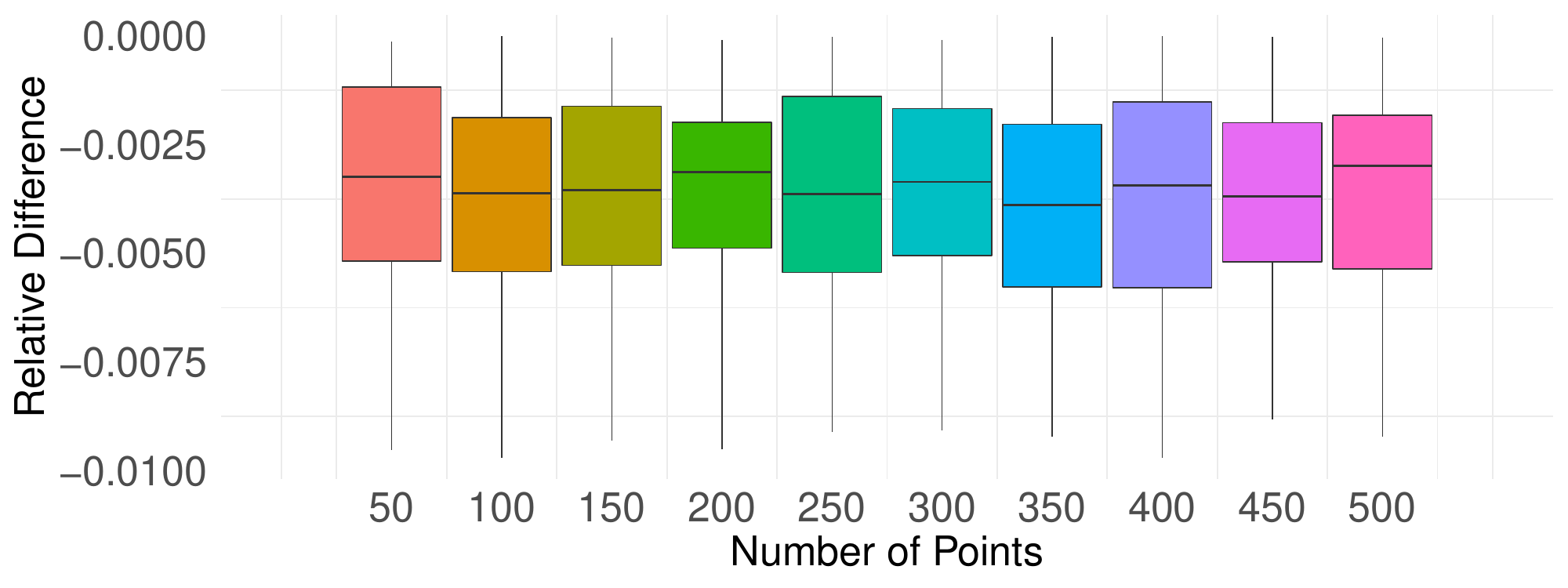}
\caption{\textsc{Hera}.}
\label{reldiff1}
\end{subfigure}
\begin{subfigure}{0.5\textwidth}
  \centering
\includegraphics[width = \textwidth, height = 0.3\textwidth]{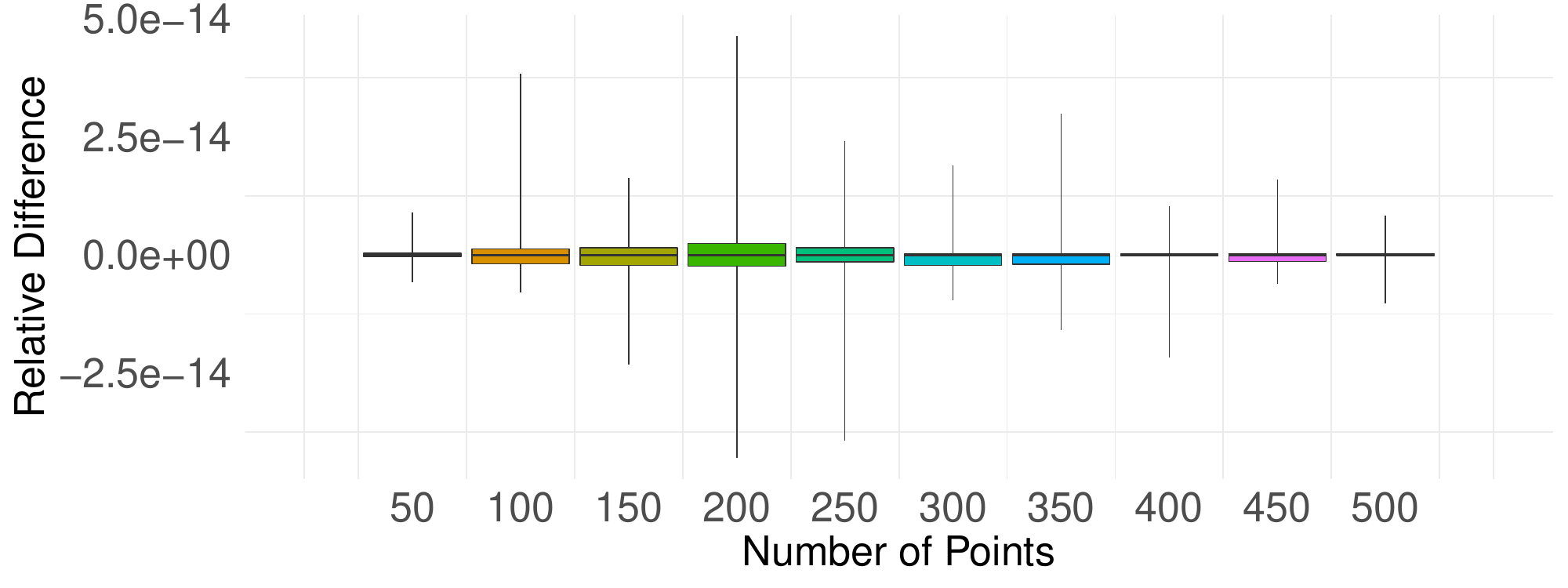}
\caption{\textsc{Persim}.}
\label{reldiff2}
\end{subfigure}
\begin{subfigure}{0.5\textwidth}
  \centering
\includegraphics[width = \textwidth, height = 0.3\textwidth]{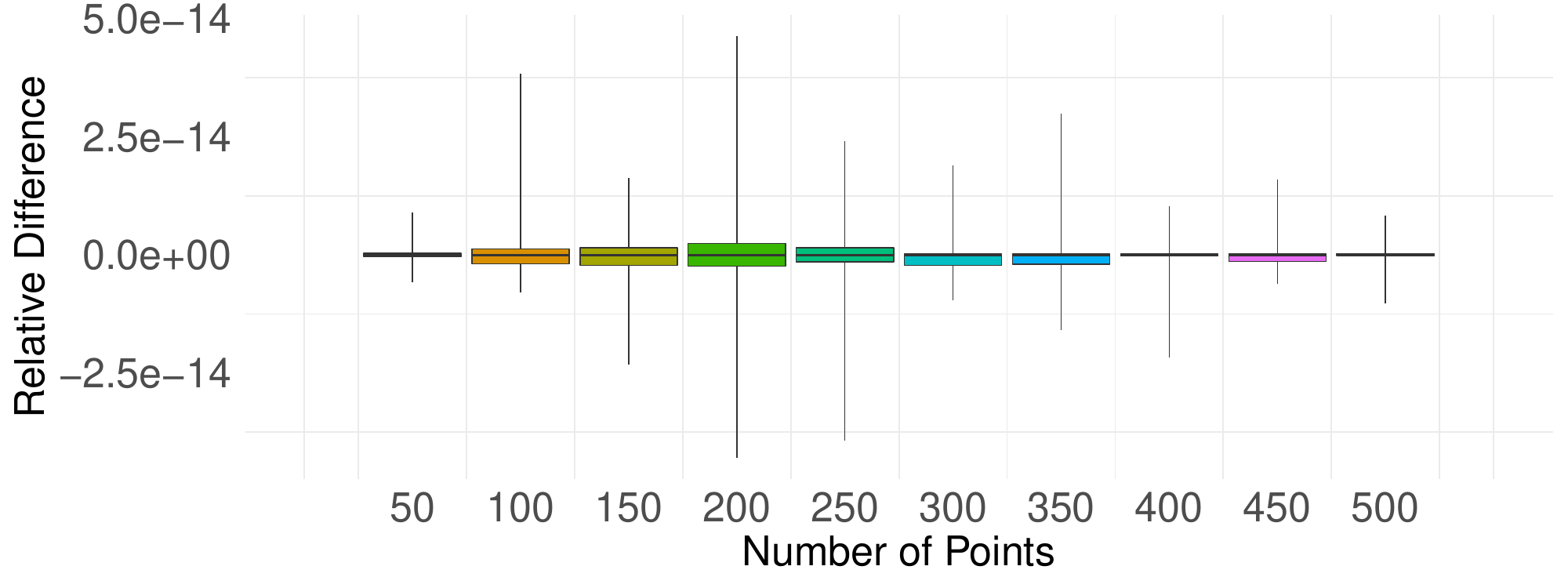}
\caption{\textsc{Lum\'awig}$_{\textsc{Py}}$.}
\label{reldiff3}
\end{subfigure}
\begin{subfigure}{0.5\textwidth}
  \centering
\includegraphics[width = \textwidth, height = 0.3\textwidth]{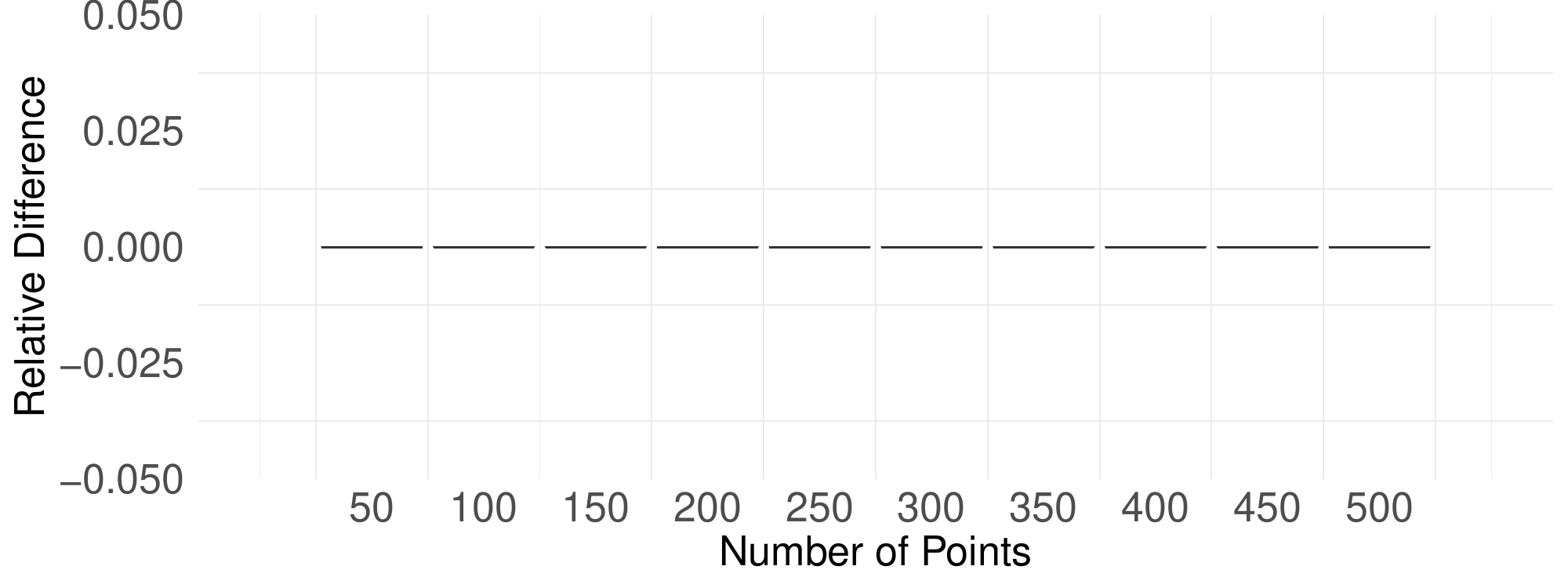}
\caption{\textsc{Lum\'awig}$_{\textsc{R}}$.}
\label{reldiff4}
\end{subfigure}
\caption{Boxplots of relative differences of the bottleneck computation output of the indicated implementation from \textsc{Dionysus}.}
\label{olddiff}
\end{figure}

We also compare the computed bottleneck distance against the output of \textsc{Dionysus}. Relative differences of the outputs of three other implementations from \textsc{Dionysus} are computed for all 100 pairs of persistence diagrams. From these bottleneck computations, descriptive summaries are obtained and plotted as boxplots in Figure \ref{olddiff}. Is it worth noting that \textsc{Hera} consistently overestimates the zero dimensional bottleneck distance relative to \textsc{Dionysus} as seen in Figure \ref{reldiff1}. Another important observation is that the output of \textsc{Lum\'awig}$_{\textsc{Py}}$ recovers that of \textsc{Persim} at a much less computational running time. Finally, we highlight that \textsc{Lum\'awig}$_{\textsc{R}}$ recovers the exact output values of the original implementation in \textsc{Dionysus}.

We remark here that while this stage of benchmarking is indeed confined within extremely small data sets, this situation in fact represents what is currently accessible to most researchers, and highlights the clear computational obstruction in the use of persistence diagrams and bottleneck distance in applications with currently available algorithms.

\subsection{Benchmarking \textsc{Lum\'awig} on larger data sets}
We perform a second stage of benchmarking against the current state-of-the-art implementation of the bottleneck distance in \textsc{Hera}. We again simulate 100 0-dimensional persistence diagrams and pair each diagram with another simulated diagram not necessarily having the same number of points. Similar descriptive measures as in the earlier stage are considered from the 100 pairs of diagrams with increasing number of points from 1,000 to 30,000. The choice of benchmarking bottleneck computation to at most 30,000 points was to draw comparison with that of \textsc{Hera} in \cite{Kerber}. One difference in our benchmarking is that the number of points in the two diagrams we are comparing need not be equal. 

\begin{figure}[h]
 \centering
\includegraphics[width = \textwidth, height = 0.3\textwidth]{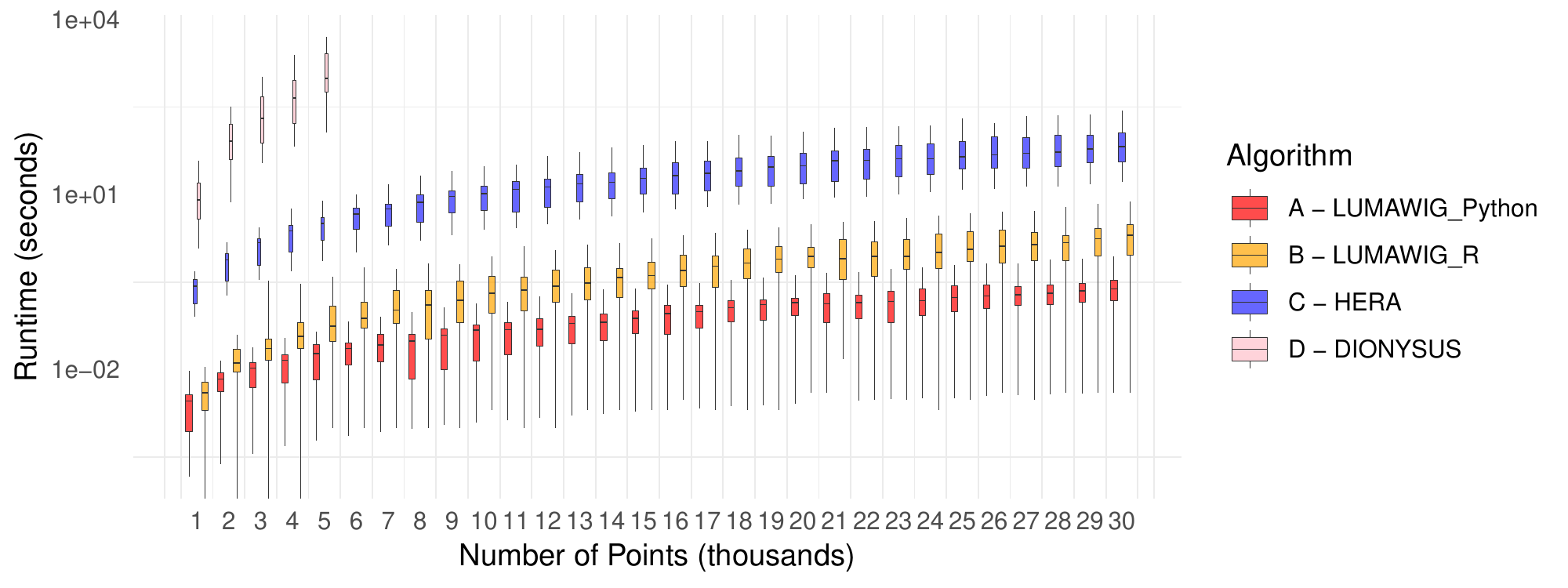}
\caption{Running time (seconds in $\log$ scale) of \textsc{Lum\'awig} versus the current state-of-the-art implementation in \textsc{Hera}. Five boxplots for the running time of the original algorithm in \textsc{Dionysus} are superimposed for reference.}
\label{newrun}
\end{figure}

Figure \ref{newrun} shows the running time distribution of 100 dimension zero bottleneck distance computations over increasing diagram sizes. Note that the vertical axis is displayed in logarithmic scale. Only five boxplots for the running time of the original algorithm implemented in \textsc{Dionysus} are superimposed to provide reference for the state-of-the-art \textsc{Hera} and our two implementations of \textsc{Lum\'awig}. A quick inspection reveals that both implementations of \textsc{Lum\'awig} are consistently several orders of magnitude faster than the current state-of-the-art \textsc{Hera}. The use of the same pairs of simulated persistence diagrams for bottleneck computations across implementations allowed for paired tests of significant difference in running time relative to \textsc{Hera}. These significant values, computed at $\alpha=0.95$ level of confidence, appear in Table \ref{tab1}.

\begin{table}[h]
\caption{Summary of significant decrease (at confidence level $\alpha=0.95$) in running time (in seconds) for paired tests versus \textsc{Hera}. Column labels are in thousands of points.}
\centering
\label{tab1}

\begin{tabular}{lcccccccccc}
\toprule
&1&2&3&4&5&6&7&8&9&10\\
\midrule
%&&&&&&&&&&\\
$\textsc{Lum\'awig}_{\textsc{R}}$&0.230&0.651&1.203&1.971&2.744&3.957&5.003&6.813&8.348&9.983\\
%&&&&&&&&&&\\
$\textsc{Lum\'awig}_{\textsc{Py}}$&0.231&0.659&1.221&2.004&2.797&4.037&5.107&6.928&8.498&10.181\\
%&&&&&&&&&&\\
\midrule
&11&12&13&14&15&16&17&18&19&20\\
\midrule
%&&&&&&&&&&\\
$\textsc{Lum\'awig}_{\textsc{R}}$&11.029&12.227&14.985&15.733&18.983&21.588&23.580&26.801&29.425&33.316\\
%&&&&&&&&&&\\
$\textsc{Lum\'awig}_{\textsc{Py}}$&11.255&12.483&15.296&16.078&19.410&22.087&24.124&27.438&30.153&34.129\\
%&&&&&&&&&&\\
\midrule
&21&22&23&24&25&26&27&28&29&30\\
\midrule
%&&&&&&&&&&\\
$\textsc{Lum\'awig}_{\textsc{R}}$&38.555&39.818&41.080&44.324&49.933&54.441&57.183&60.196&66.510&72.948\\
%&&&&&&&&&&\\
$\textsc{Lum\'awig}_{\textsc{Py}}$&39.427&40.734&42.073&45.407&51.129&55.725&58.517&61.605&68.200&74.879\\
%&&&&&&&&&&\\
\bottomrule
\end{tabular}
\end{table}

\begin{figure}[h]
\begin{minipage}{0.45\textwidth}%\
  \centering
    \begin{subfigure}{\textwidth}
  \centering
\includegraphics[width =  \textwidth, height = 0.5\textwidth]{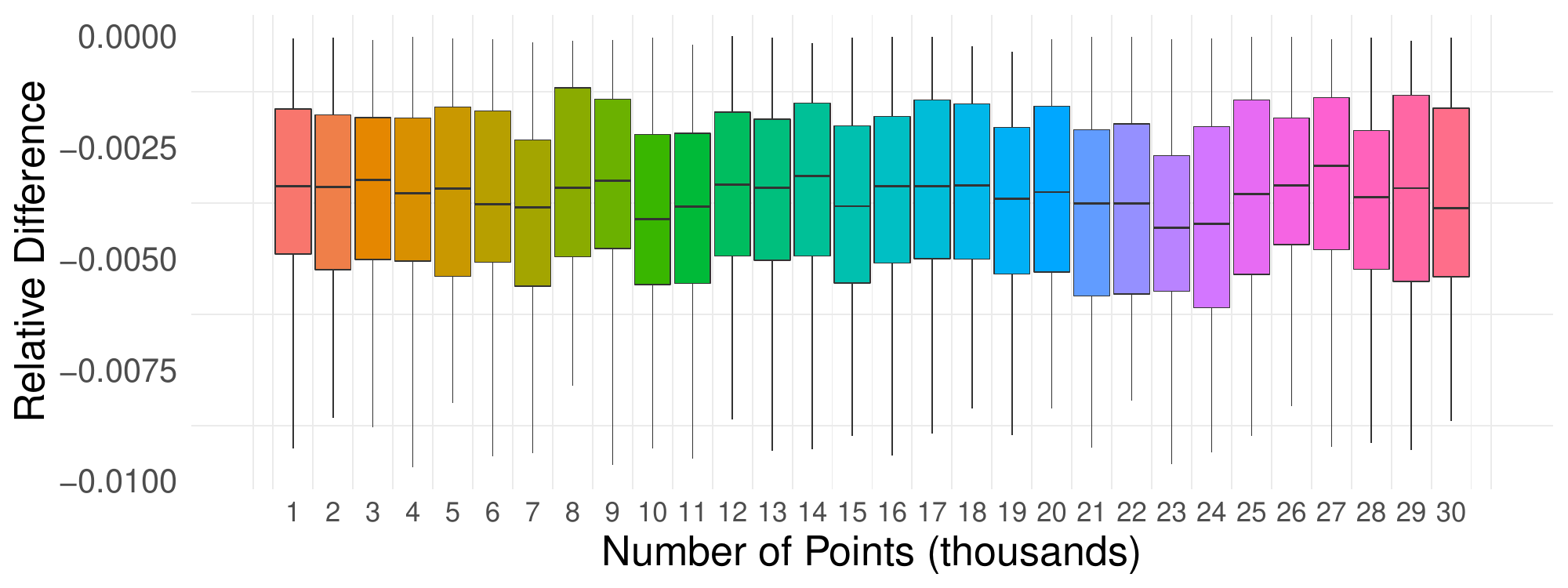}
\caption{\textsc{Hera} versus \textsc{Lum\'awig}$_\textsc{R}$.}
\label{reldiffDio}
\end{subfigure}
    \begin{subfigure}{\textwidth}
    \vspace{10pt}
  \centering
\includegraphics[width = \textwidth, height = 0.5\textwidth]{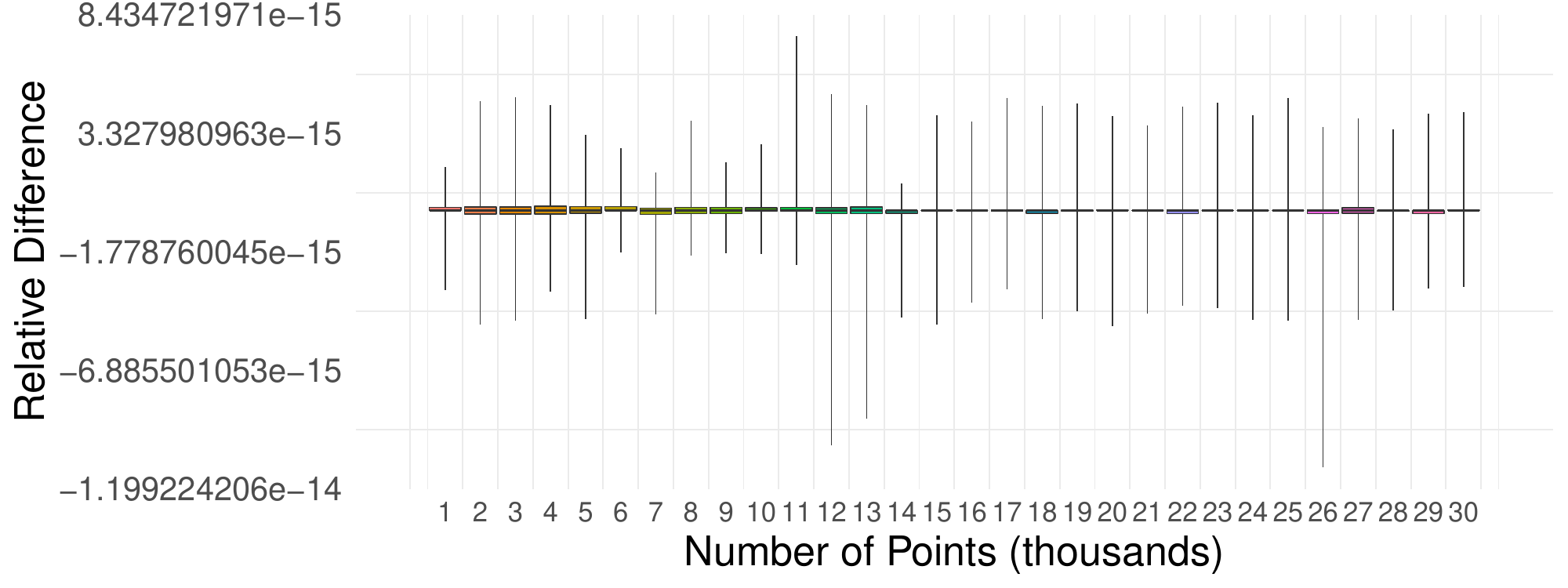}
\caption{\textsc{Lum\'awig}$_\textsc{Py}$ versus \textsc{Lum\'awig}$_\textsc{R}$.}
\label{reldiffByboPy}
\end{subfigure}
    \end{minipage}%
      \centering
\begin{minipage}{0.55\textwidth}%
\begin{subfigure}{\textwidth}
 \centering
\includegraphics[trim= 0 65 20 0,clip,width = \textwidth, height = \textwidth]{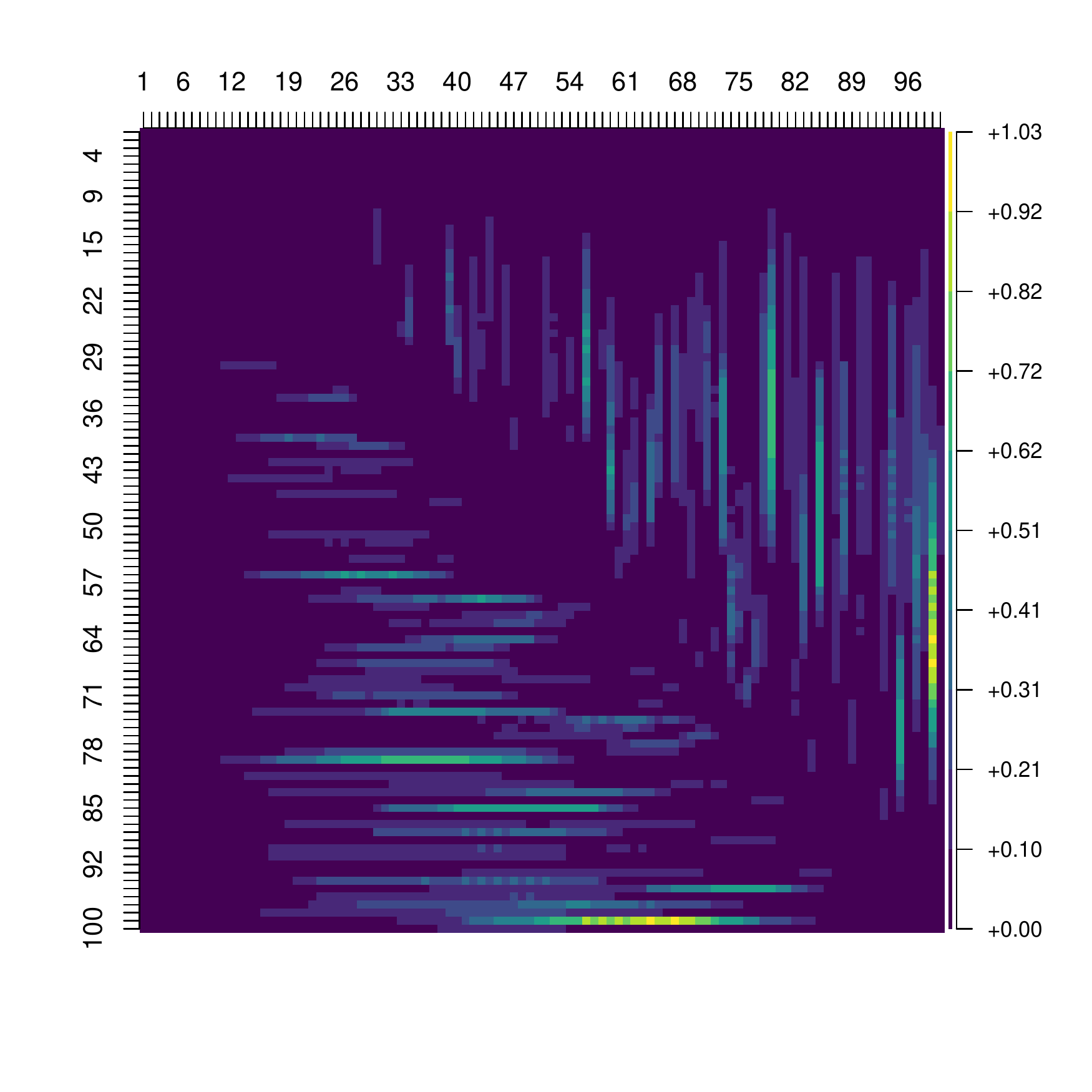}
\caption{Heatmap.}
\label{heatmap}
\end{subfigure}
\end{minipage}%
\caption{(a)-(b) Boxplots of relative differences between the bottleneck computation outputs of the indicated pair of implementations. (c) Heat map of the median running times of \textsc{Lum\'awig}$_{\textsc{R}}$. Each pixel represents the median running time (in seconds) for 100 computations of dimension zero bottleneck distance between diagrams. The number of points in the diagrams are in units of 1000.}
\label{botdiffall}
\end{figure}

As \textsc{Lum\'awig}$_\textsc{R}$ yields exact values for the bottleneck distance relative to the original \textsc{Dionysus} implementation, we use it as basis in the computation of relative differences in this stage. Figures \ref{reldiffDio} and \ref{reldiffByboPy} show the relative difference in the computed dimension zero bottleneck distance respectively of \textsc{Hera} and \textsc{Lum\'awig}$_\textsc{Py}$ with respect to that of \textsc{Lum\'awig}$_\textsc{R}$. Consistent with the comparison between the outputs of \textsc{Hera} and \textsc{Dionysus} in Figure \ref{reldiff1}, \textsc{Hera} consistently overestimates the dimension zero bottleneck distance with respect to that of \textsc{Lum\'awig}$_\textsc{R}$. In contrast, relative differences between the two implementations of \textsc{Lum\'awig} can be attributed to rounding differences between \textsc{Python} and \textsc{R}.

\subsection{Complexity analysis}
Figure \ref{heatmap} shows a heat map of the median running time of \textsc{Lum\'awig}$_\textsc{R}$ over 100 computations per pixel of the bottleneck distance between pairs of persistence diagram with varying number of points: for $i\leq j$, pixel $(i,j)$ represents the median running time for the computation of the bottleneck distance between a diagram with $i$ thousands of points and another diagram with $j$ thousands of points, such that each set of points has death times uniformly chosen from the interval range $(0,2000i)$ and $(0,2000j)$ respectively. It can be inferred from this figure that the best running times happen along the main diagonal, as well as the upper and left portions of the heat map. These correspond to two specific cases: when the diagrams have equal number of points, or when one diagram has overwhelmingly more points than the other. In contrast, regions in the heat map that show increased running times correspond to the case when a diagram that has a large number of points is compared to another that has about half as many points. This observation is supported in the next figure.
\begin{figure}[t]
 \centering
\includegraphics[width = \textwidth, height=0.63\textheight]{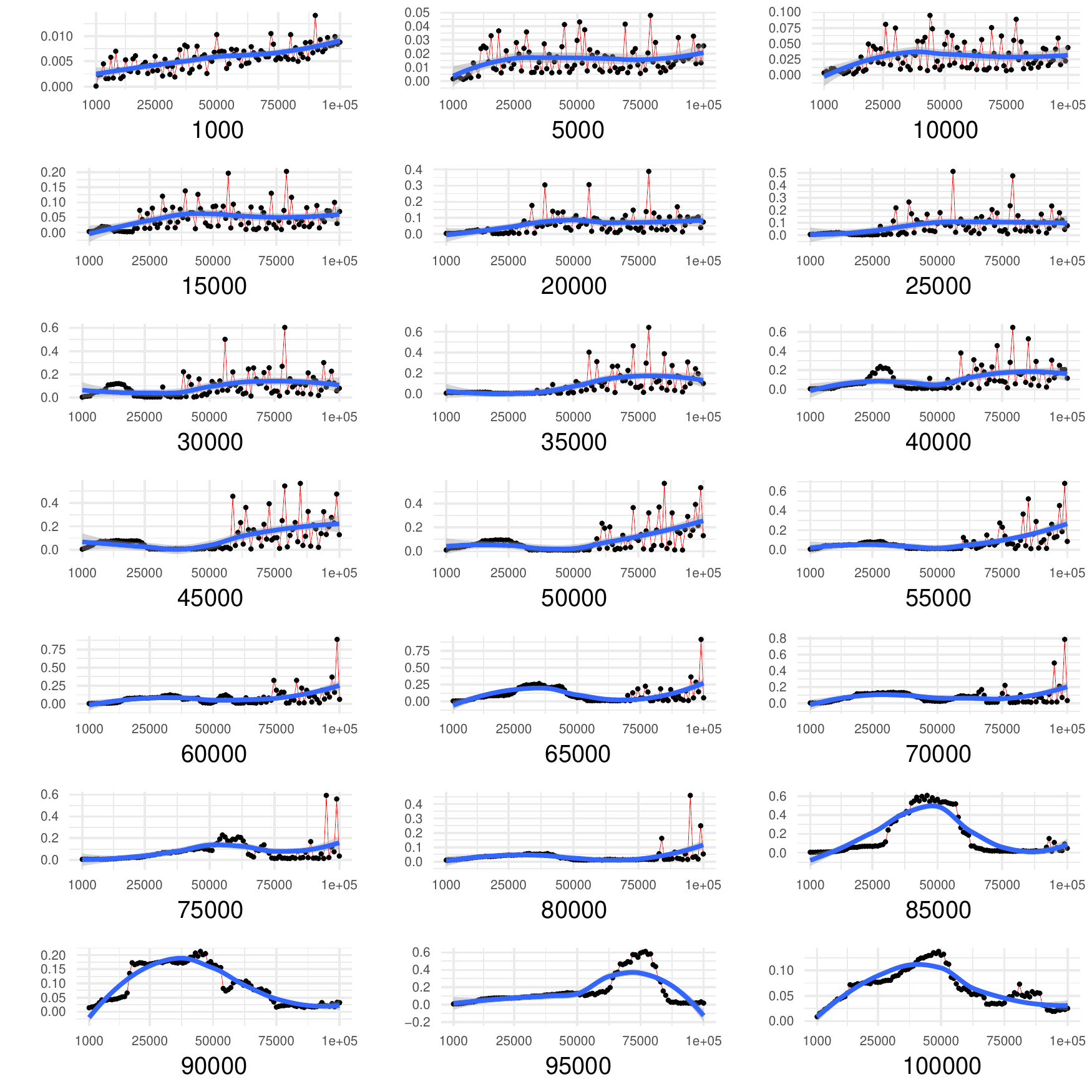}
\caption{Scatter plots with fitted curves of the median running times (in seconds) of \textsc{Lum\'awig}$_\textsc{R}$ over 100 computations of dimension zero bottleneck distance between a base diagram with the labeled number of points and a diagram with $k$ thousands of points, for $k=1,2,...,100$.}
\label{graphs}
\end{figure}

Figure \ref{graphs} shows several scatter plots with fitted curves of the median running times of \textsc{Lum\'awig}$_\textsc{R}$ over 100 computations of dimension zero bottleneck distance. The label in each scatterplot represents the number of points in the fixed base diagram, and a point in the scatterplot at the $k$ thousand mark along the horizontal axis represents the median running time over 100 computations of the bottleneck distance between the base diagram and another diagram with $k$ thousand points whose death times are uniformly chosen from the interval range $(0,2k)$.

To further investigate the observations above, we examine the performance of \textsc{Lum\'awig}$_\textsc{R}$ in the computation of dimension zero bottleneck distance in four pairs of settings for size of the diagrams and the range of values the death times are drawn from. The first is when \textsc{Lum\'awig}$_\textsc{R}$ is tasked to compare two persistence diagrams with the same number of points whose death times are drawn from the same range of values. We calculate the dimension zero bottleneck distance over 100 pairs of persistence diagrams of equal sizes starting from 1,000 points to 1,000,000 points. Every diagram is simulated in the same manner as the previous experiments. Median running times are then plotted and fitted with a regression curve. Midspread and range for every 100 computations at every unit of 1,000 points are superimposed to illustrate the distribution of running times. Figure \ref{equalboth} shows an excellent linear fit ($R^2 = 0.99$) for the running time. We also highlight the observed experimental result that the running time between two diagrams each having 1 million points with death times drawn from the range $(0, 2000000)$ averages to between 2 and 3 tenths of a second.

The second setting involves two diagrams of the same size but the range of death values for the second is half as wide as the first. In this case, we see in Figure \ref{equalsizehalfrange} that the running time trend is perfectly fitted with a linear curve. The third setting considers two diagrams where the second has half as many points as the first. We remark that this setting differs from that performed for Figure \ref{heatmap} in that the range where the death times are drawn from for the simulated diagrams in this experiment is the same for the two diagrams. We do this to ensure that any observed significant difference in performance is attributable only to fixed difference in the number of points between the diagrams. As we observe an increased running time for  \textsc{Lum\'awig}$_\textsc{R}$ in this case, we compute only to until there are 100,000 points in the larger diagram. Figure \ref{halfsizeequalrange} shows two fitted regression curves: a quadratic fit with $R^2 = 1$ and a linear fit with $R^2 = 0.95$. We highlight that even for the case where \textsc{Lum\'awig}$_\textsc{R}$ evidently takes longer to compute the dimension zero bottleneck distance, a linear model provides a very good fit for the trend.

\begin{figure}
\begin{subfigure}{\textwidth}
 \centering
\includegraphics[width = \textwidth, height = 0.2\textwidth]{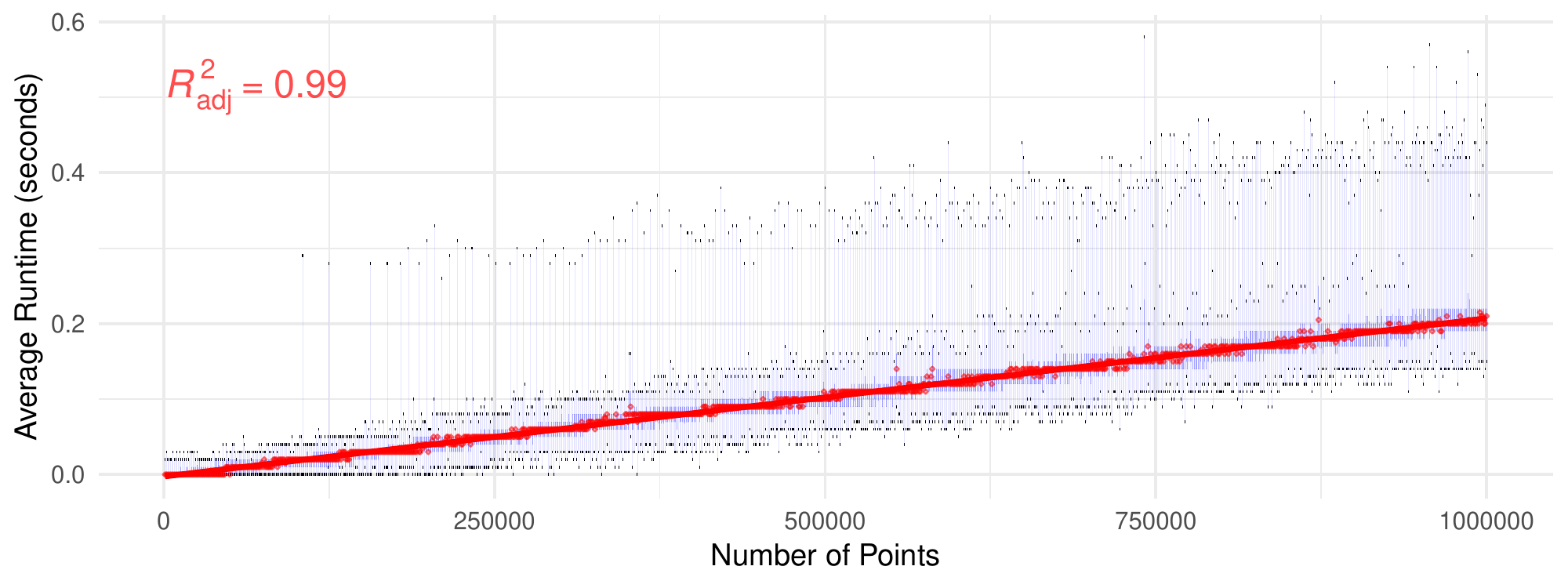}
\caption{Equal size and range.}
\label{equalboth}
\end{subfigure}\\
\begin{subfigure}{\textwidth}
 \centering
\includegraphics[width = \textwidth, height = 0.2\textwidth]{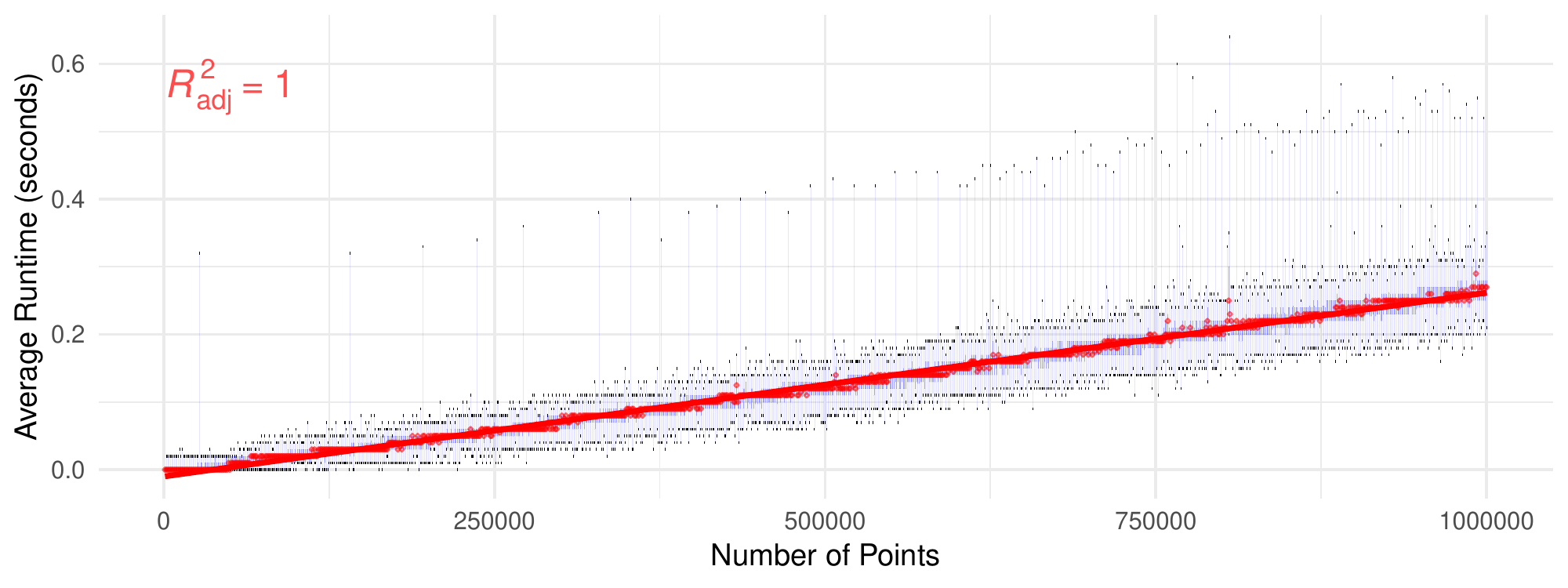}
\caption{Equal size but different range.}
\label{equalsizehalfrange}
\end{subfigure}\\
\begin{subfigure}{\textwidth}
 \centering
\includegraphics[width = \textwidth, height = 0.2\textwidth]{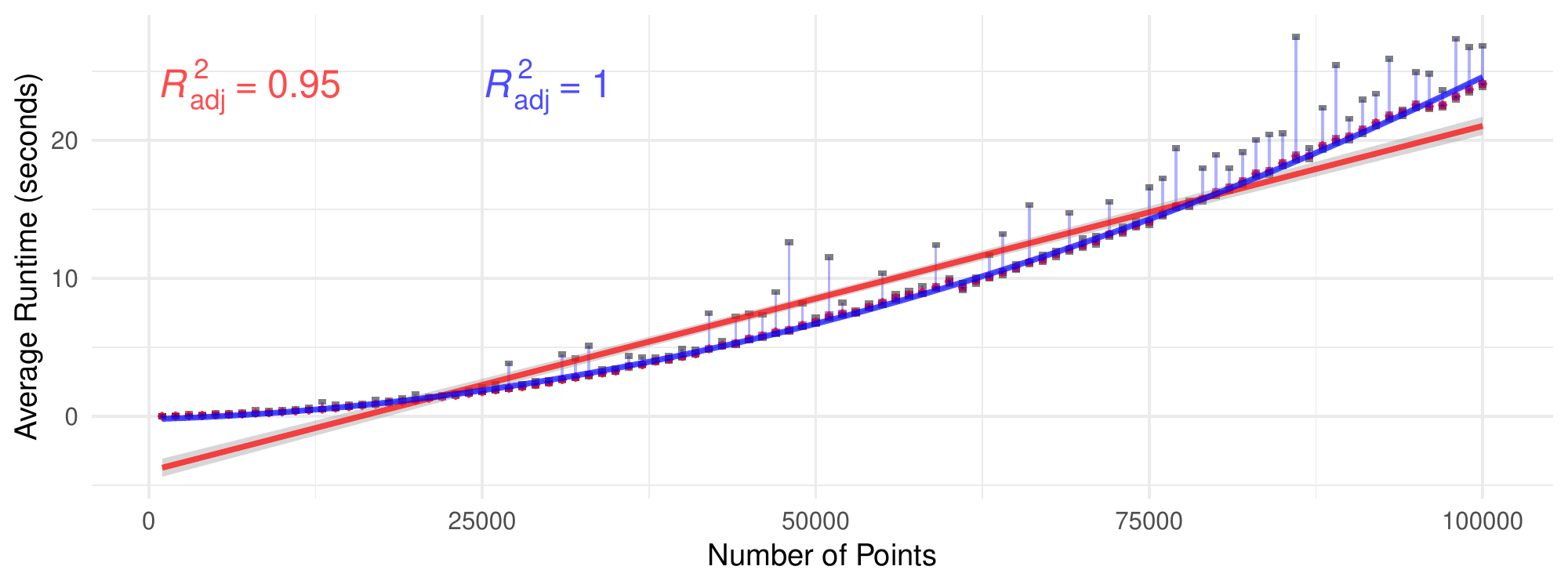}
\caption{Different size but equal range.}
\label{halfsizeequalrange}
\end{subfigure}\\
\begin{subfigure}{\textwidth}
 \centering
\includegraphics[width = \textwidth, height = 0.2\textwidth]{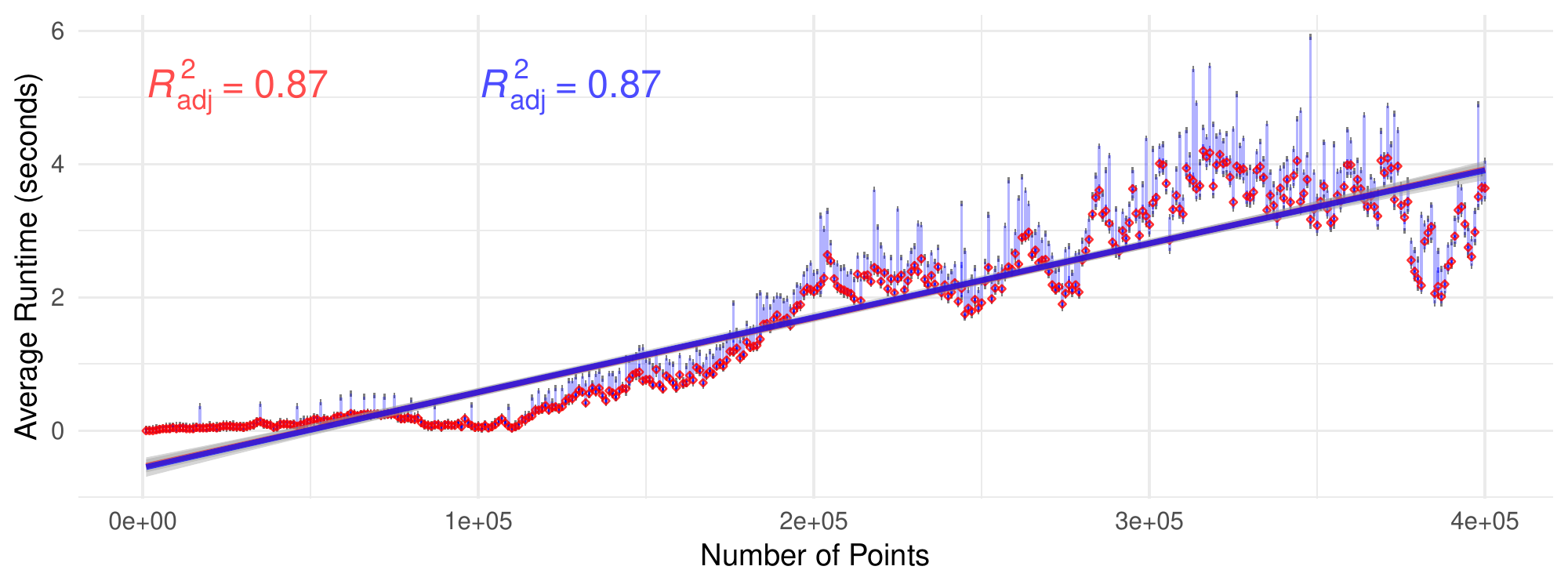}
\caption{Different size and range.}
\label{halfsizehalfrange}
\end{subfigure}\\
\caption{Median running time in the computation of bottleneck distance between two diagrams with varying size and range settings fitted with regression curves. Superimposed are the minimum and maximum running times over the 100-run simulation per unit of 1000 points to illustrate the running time range, and the narrow darker blue band to show the midspread.}
\label{linear}
\end{figure}

The final setting is where the second of two diagrams has half as many points with death values drawn from a range half as wide as that for the first. Regression curves are again shown in Figure \ref{halfsizehalfrange} with linear and quadratic fit both achieving $R^2 = 0.87$.

\section{\textsc{Lum\'awig} in Digit Classification}
With new access to a fast algorithm for computing dimension zero bottleneck distance, we leverage persistence and other clustering-based diagrams to craft features for digit classification. We classify 10,000 28$\times$28-pixel digit images in the MNIST data set via a random forest classifier. Similar to Garin and Tauzin \cite{mnist}, we train the classifier using features based on topological summaries. However, we depart from Garin and Tauzin's approach in that we only extract features from dimension zero persistence diagrams and other related clustering-based diagrams. In particular, we craft statistical summaries from distributions of bottleneck distances computed from diagrams resulting from dimension zero persistent homology and clustering of multiple sub-collections of points. We summarize this procedure next. For a detailed account of this procedure, we point the interested reader to \cite{ignacio2}  where it is used to recover higher dimensional shape information of digits from intrinsic clustering behavior. 

The first step in the procedure is to generate multiple collections of points from the digits via samples extracted based on point distributions referenced from nine pre-selected landmark points. We use the same landmark points introduced in \cite{mnist}. Sampling is also done across multiple resolutions by varying the number of points selected in every bin of every distribution histogram. Then,  for each sampled sub-collection of points in each  sampling resolution, persistent homology and clustering algorithms are respectively used to generate persistence and clustering diagrams. We gather diagrams by their sampling setting and compute pairwise bottleneck distances using \textsc{Lum\'awig}. Finally, we compute statistical summaries from the distributions of computed bottleneck distances, and use these to train a random forest classifier with 1000 trees. 

We perform a 10-fold cross validation on our training set of 10,000 digit images from MNIST, and report the summary of obtained $F_1$ scores in Table \ref{tab2}. The average class predictions of the random forest are summarized in the confusion matrix in Figure \ref{confusion}.
\begin{table}[h]
\caption{Summary of $F_1$ scores on a 10-fold cross validation.}
\centering
\label{tab2}

\begin{tabular}{lccccccccccc}
\toprule
Digit&0&1&2&3&4&5&6&7&8&9&Overall\\
\midrule
Mean&0.841&0.940&0.678&0.727&0.687&0.709&0.847&0.754&0.745&0.754&0.768\\
Std. Dev.&0.030&0.011&0.033&0.019&0.032&0
.037&0.030&0.020&0.015&0.032&0.011\\
\bottomrule
\end{tabular}
\end{table}

\begin{figure}
 \centering
\includegraphics[width = 0.35\textwidth, height = 0.38\textwidth]{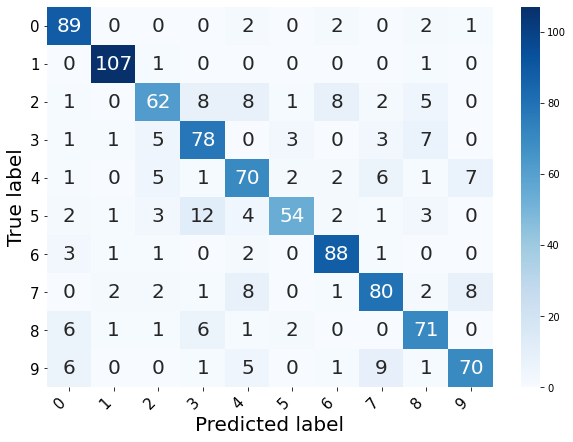}
\caption{Confusion matrix for the average prediction of the random forest over a 10-fold cross validation.}
\label{confusion}
\end{figure}

The results above show that the random forest classifier is able to use our crafted bottleneck-based features to classify, at a respectable level of accuracy, the 10 digits in the MNIST data set despite all digits possessing the same dimension zero topological signature of having only one connected component. In particular, we infer from the exceptionally high score on the classification of the simplest digit 1, that differences captured by the bottleneck distance in the clustering behavior across multiple point samples of this digit is outstandingly subtle, and hence different, from the rest. 
\section{Discussions and Conclusion}

Our benchmarking experiments reveal that \textsc{Lum\'awig} outperforms, by several orders of magnitude, all currently available implementations of dimension zero bottleneck distance in terms of running time. \textsc{Lum\'awig} also recovers the exact bottleneck distance produced by \textsc{Dionysus}. We believe this is a significant contribution as it affords a viable tool to process and utilize dimension zero persistence diagrams in comparing evolving connectivity information between large data sets in a manner that goes beyond the simple use of the most persistent components. Even now, a truly comprehensive and holistic treatment of information embedded in dimension zero persistence diagrams has been left unexplored due primarily to the lack of feasible machinery that can handle significant scaling up in data size. In fact, this note presents the first instance that the bottleneck distance is used in practice for data of magnitude and scale in the order of up to a million. In particular, we see that \textsc{Lum\'awig} only takes an average of 2 to 3 tenths of a second to compute the bottleneck distance between diagrams each having one million points.

A natural question to ask is whether a similar strategy of methodically modifying a specific initial bijection to recover all possible cases that yield the best matching for the general case, where birth times of features need not be at the beginning of the filtration (this covers the bottleneck distance for higher dimensional features) is possible. We note that an important first step is to induce an appropriate partial order on the points in each diagram that can accommodate a case-exhaustive approach to optimize the norm. Moreover, the added degree of freedom will naturally introduce cases we have not considered in our optimization step.

Our empirical tests suggest that \textsc{Lum\'awig} enjoys linear complexity for the case where both diagrams have equal number of points. Moreover, we also see that even for the special case revealed by Figure \ref{graphs}, where there is an apparent slowdown in computational time, the trend seen when data size scales up is also practically linear (see Figures \ref{halfsizeequalrange} and \ref{halfsizehalfrange}). In a future note, we plan to provide a more comprehensive analysis for complexity. Nevertheless, we are confident that \textsc{Lum\'awig} can be useful in practical applications of TDA at this stage. 

Finally, our application on digit classification showcases, in the same significant way as Weber et al. did in \cite{traffic}, the potential in leveraging persistence diagrams and bottleneck distance as sources of novel features for machine learning tasks. It is our hope that \textsc{Lum\'awig} contributes in paving the way for this direction in TDA research. 

\section{Repository for \textsc{Lum\'awig}}
A repository for \textsc{Lum\'awig} will eventually be set up and maintained as soon as licences, copyright certificates, and other clearances are secured. 

%%%%%%%%%%%%%%%%%%%%%%%%%%%%%%%%%%%%%%%%%%
\section{Acknowledgement}
P.S.I. would like to thank the University of the Philippines Baguio for the research load credit he received to conduct this work.

% References, variant A: internal bibliography
%=====================================

\end{document}